\documentclass[conference]{IEEEtran}
\IEEEoverridecommandlockouts
\usepackage{algorithmic}
\usepackage{amsfonts,amsmath,amssymb,amsthm}

\usepackage{bm}
\usepackage{booktabs}
\usepackage{calligra}
\usepackage{caption}
\usepackage[authormarkup=none, final]{changes}
\usepackage{subcaption}
\usepackage{changepage}
\usepackage{cite}
\usepackage{cleveref}
\usepackage{cuted}
\usepackage{eucal}
\usepackage{float}
\usepackage[left=0.667in,
            right=0.667in,
            top=0.792in,
            bottom=0.597in]{geometry}
\usepackage{graphicx}
\usepackage{lineno}
\usepackage{longtable}
\usepackage{mathabx}
\usepackage{mathtools}
\usepackage{multirow}
\usepackage{optidef}
\usepackage{pgfplots}
\usepackage{siunitx}
\usepackage{soul}
\usepackage{tabu}
\usepackage{tabularx}
\usepackage{textcomp}
\usepackage{tikz}
\usepackage{wasysym}

\def\BibTeX{{\rm B\kern-.05em{\sc i\kern-.025em b}\kern-.08em
    T\kern-.1667em\lower.7ex\hbox{E}\kern-.125emX}}
    \renewcommand{\vec}[1]{\mathbf{#1}}

\newtheorem{problem}{Problem}
\newtheorem{remark}{Remark}
\newtheorem{proposition}{Proposition}
\newcommand{\mat}[1]{\mathbf{#1}}

\DeclareMathAlphabet{\mathcalvar}{OMS}{cmsy}{m}{n}

\newcount\Comments  % 1 suppresses notes to selves in text
\newcommand{\kibitz}[2]{\ifnum\Comments=0\textcolor{#1}{#2}\fi}

\definechangesauthor[name={Fangyu Wu}, color=red]{FW}
\colorlet{Changes@Color}{red}

\usetikzlibrary{decorations.pathreplacing,angles,quotes}
\usetikzlibrary{arrows.meta}
\usetikzlibrary{shapes.geometric}
\usepgfplotslibrary{fillbetween}
\usepgfplotslibrary{colorbrewer}
\pgfplotsset{compat=1.16}
\pgfplotsset{yticklabel style={text width=2em,align=right}}
\pgfplotsset{every tick label/.append style={font=\small}}
\pgfdeclareplotmark{truestar}{
    \node[star,star point ratio=2.25,minimum size=6pt,
          inner sep=0pt,draw=olive,solid,fill=olive] {};
}

\begin{document}

\title{A Hierarchical MPC Approach to Car-Following\\ via Linearly Constrained Quadratic Programming\\
\thanks{This material is based upon work supported by the National Science Foundation under Grant Numbers CNS-1837244.
Any opinions, findings, and conclusions or recommendations expressed in this material are those of the author(s) and do not necessarily reflect the views of the National Science Foundation.
This material is based upon work supported by the U.S. Department of Energy’s Office of Energy Efficiency and Renewable Energy (EERE) award number CID DE-EE0008872.
The views expressed herein do not necessarily represent the views of the U.S. Department of Energy or the United States Government.}
}

\author{\IEEEauthorblockN{Fangyu Wu}
\IEEEauthorblockA{\textit{Dept. of EECS} \\
\textit{University of California, Berkeley}\\
Berkeley, USA \\
fangyuwu@berkeley.edu}
\and
\IEEEauthorblockN{Alexandre M. Bayen}
\IEEEauthorblockA{\textit{Dept. of EECS} \\
\textit{University of California, Berkeley}\\
Berkeley, USA \\
bayen@berkeley.edu}
}

\maketitle

\begin{abstract}
Single-lane car-following is a fundamental task in autonomous driving.
A desirable car-following controller should keep a reasonable range of distances to the preceding vehicle and do so as smoothly as possible.
To achieve this, numerous control methods have been proposed:
some only rely on local sensing; others also make use of non-local downstream observations.
While local methods are capable of attenuating high-frequency velocity oscillation and are economical to compute,
non-local methods can dampen a wider spectrum of oscillatory traffic but incur a larger cost in computing.
In this article, we design a novel non-local tri-layer MPC controller that is capable of smoothing a wide range of oscillatory traffic and is amenable to real-time applications.
At the core of the controller design are \textsl{1}) an accessible prediction method based on ETA estimation and \textsl{2}) a robust, light-weight optimization procedure, designed specifically for handling various headway constraints.
Numerical simulations suggest that the proposed controller can simultaneously maintain a variable headway while driving with modest acceleration and is robust to imperfect traffic predictions.
\end{abstract}


\section{Introduction}
\label{sec:introduction}
Despite significant progress in recent years, autonomous driving still remains a challenging problem.
Among many problems in this area, single-lane car-following is arguably a fundamental quest.

It is widely known that human car-following is inherently sub-optimal.
% First coined in~\cite{swaroop1996string}, this unstable oscillatory behavior is referred to as \textit{string instability}.
Suigiyama and Tadaki have demonstrated in~\cite{sugiyama2008traffic,tadaki2013phase} that stop-and-go waves could emerge in a circular platoon of human vehicles without the presence of physical bottlenecks or lane changes.

To improve, various \textit{local} controllers have been proposed~\cite{martinez2007safe,kesting2008adaptive,luo2010model,kerner2018physics,stern2018dissipation}.
For example, Luo has described a model predictive control (MPC) formulation for fuel-efficient, fixed-headway adaptive cruise control~\cite{luo2010model}.
Kerner et. al. have presented a variable-headway controller that can attenuate propagation of stop-and-go waves~\cite{kerner2018physics}.
In addition, Stern et. al. have demonstrated that a single vehicle controlled by a proportional-integral controller is able to dampen stop-and-go waves induced by human drivers in a circular platoon~\cite{stern2018dissipation}.
% A common feature of the above controllers is that they achieve smoother,  more efficient driving using only local observations.

Beyond local methods, various approaches on \textit{non-local} controllers have been proposed in an effort to attenuate a wider spectrum of oscillations, such as~\cite{bu2010design,schmied2015nonlinear,jia2018energy}.
The key advantage of non-local controllers lies at the fact that with non-local knowledge of downstream traffic it can react in advance to dampen non-local oscillations.
% \added[id=FW]{A common pattern in non-local methods is a multi-layer, hierarchical structure, commonly including a higher-level planner and a lower-level controller.}
For instance, \cite{schmied2015nonlinear,jia2018energy} introduced various variable-headway MPC methods to achieve smoother car-following by using information from downstream traffic lights.
% California PATH has described a cooperative adaptive cruise controller and demonstrated its string stability~\cite{bu2010design}.
% Jia has introduced three variations of~Schmied with different performance trade-offs~\cite{jia2018energy}.

\added[id=FW]{
A common feature of the non-local methods above is a bi-layer structure consisting of a prediction layer and a control layer.
To predict, these methods require dedicated sensing infrastructure and either vehicle-to-infrastructure~\cite{schmied2015nonlinear,jia2018energy} or vehicle-to-vehicle~\cite{bu2010design} communication.
While non-local methods may have higher performance upper bound, they are also more prone to errors~\cite{hyeon2019influence}.
Yet, mitigating the impacts of erroneous predictions has been largely overlooked in previous works.

To address prediction errors, a common idea is to subdivide the control layer into a planning layer and a tracking layer.
First, the planning layer produces a reference trajectory using the forecast from the prediction layer.
Then, the tracking layer attempts to track this reference without violating certain performance constraints.
This approach has found many applications in transportation outside of car-following.
For example, \cite{malikopoulos2018decentralized,shivam2020intersection} have applied this idea for autonomous intersection coordination.

Adopting this idea of planning-then-tracking, we propose a novel hierarchical MPC controller, consisting of a prediction layer, a planning layer, and a tracking layer.
The core MPC formulation assumes the form of a linearly constrained quadratic program (LCQP).
Due to the tri-layer design, the controller is robust to prediction errors.
In contrast to~\cite{schmied2015nonlinear,hyeon2019short}, the method does not require dedicated sensing and communication infrastructure for prediction.
Rather, it only uses an estimated time of arrival (ETA) estimator, which is commonly accessible from mainstream map service providers and needs only cellular network.
In addition, compared to~\cite{schmied2015nonlinear,jia2018energy}, our method admits a more flexible formulation of headway constraints and is better understood when it comes to theoretical properties such as optimality and feasibility.

We regard the primary contribution of the paper as a novel tri-layer MPC design, along with a comprehensive numerical study that demonstrates its robustness to prediction errors.
At the code of the MPC design are:
\textsl{1}) A simple, accessible prediction method using an ETA estimator and
\textsl{2}) A robust, light-weight LCQP formulation.

We state the problem of the two-vehicle car-following in Section~\ref{sec:problem_statement}.
Next, we formally define and analyze the tri-layer hierarchical MPC control scheme in Section~\ref{sec:optimal_car-following_controller}.
We evaluate the performance of our MPC controller via numerical simulations in Section~\ref{sec:numerical_simulation} and discuss the numerical results in Section~\ref{sec:discussions}.
Lastly, we conclude our findings in Section~\ref{sec:conclusions}.
}

\section{Problem Statement}
\label{sec:problem_statement}

Consider a two-vehicle setup, where a considered vehicle follows a preceding vehicle that travels along a path of length $L$.
Denote the initial and final times of the trip as $t_{0}$ and $t_{f}$, respectively.

Let the state of a vehicle be $\vec{x} \coloneqq [s \quad v]^{\top}$, where $s$ and $v$ represent the position and speed of a vehicle, respectively.
Define $a$ as the acceleration of a vehicle.
Let $\ell$ be the length of a vehicle.
We use superscript to identify to which vehicle a variable belongs, where the considered vehicle is indexed as one and the preceding vehicle zero.
Thus, $\vec{x}^{1}$ refers to the state of the considered vehicle.
Let the initial position of the considered vehicle be $w_{0}$, that is, $s^{1}(t_{0}) \coloneqq w_{0}$.
Let the initial position of the preceding vehicle be $w_{1}$, that is, $s^{0}(t_{0}) \coloneqq w_{1}$.
Let $w_{f} \coloneqq w_{1} + L$.
Then, $s^{0}(t_{f}) \coloneqq w_{f}$.
The two-vehicle scenario is illustrated in Figure~\ref{fig:problem_setup}.

\begin{figure}[htbp]
  \centering
  \resizebox{0.49\textwidth}{!}{%
  \begin{tikzpicture}
    \node[xscale=-1] (considered_vehicle) at (0.75, 0.5) {\includegraphics[width=0.075\textwidth]{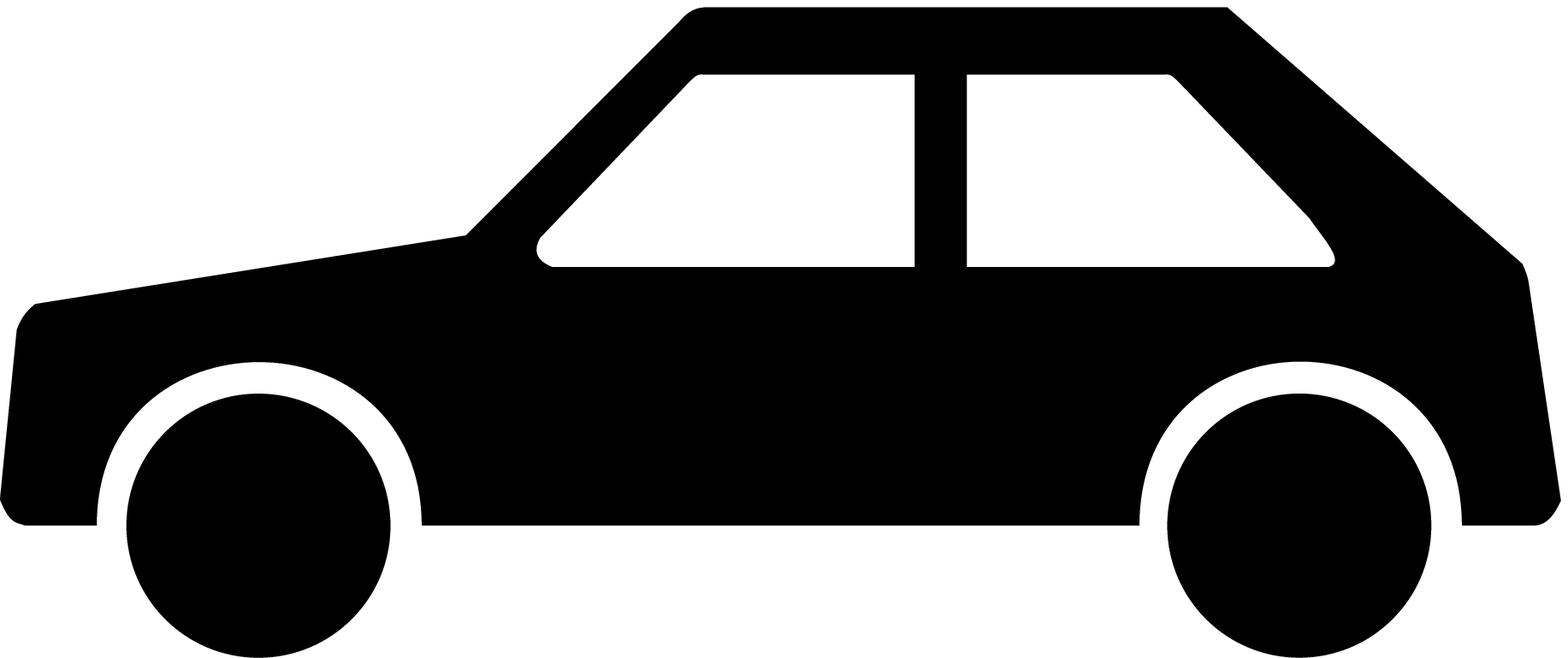}};
    \node[xscale=-1] (preceding_vehicle) at (3.5, 0.5) {\includegraphics[width=0.075\textwidth]{figures/car_pictogram.eps}};
    \draw[->, >=latex, line width=0.35mm] (0, 0) to (8, 0) node[right] {$s$};
    \draw (1.5, 0.1) to (1.5, 0) node[below] {$w_{0}$};
    \draw (4.25, 0.1) to (4.25, 0) node[below] {$w_{1}$};
    \draw (5, 0.1) to (5, 0) node[below] {$w_{2}$};
    \draw (5.75, 0) node[below] {$\dots$};
    \draw (6.5, 0.1) to (6.5, 0) node[below] {$w_{f-1}$};
    \draw (7.25, 0.1) to (7.25, 0) node[below] {$w_{f}$};
    \draw (0.75, 0.9) node[above] {$\vec{x}^{1}, a^{1}, \ell^{1}$};
    \draw (0.75, 1.5) node[above] {Consid. Veh.};
    \draw (3.5, 0.9) node[above] {$\vec{x}^{0}, a^{0}, \ell^{0}$};
    \draw (3.5, 1.5) node[above] {Preced. Veh.};
    \draw[decoration={brace}, decorate] (4.25, 0.2) -- node[above=3pt] {$\Delta{s}$} (5, 0.2);
    % \draw[decoration={brace}, decorate] (6.5, 0.2) -- node[above=3pt] {$\Delta{s}_{f-1}$} (7.25, 0.2);
    \draw[decoration={brace,mirror,raise=5pt}, decorate] (4.25, -0.25) -- node[below=6pt] {$L$} (7.25, -0.25);
    \draw[->, >=latex, line width=0.2mm, dotted] (1.5, 0.5) to node[above]{\footnotesize RADAR} (2.75, 0.5);
    \draw[->, >=latex, line width=0.2mm, dotted] (-0.75, 1.25) |- (0, 0.5);
    \node (network) at (-0.75, 1.5) {\includegraphics[width=0.05\textwidth]{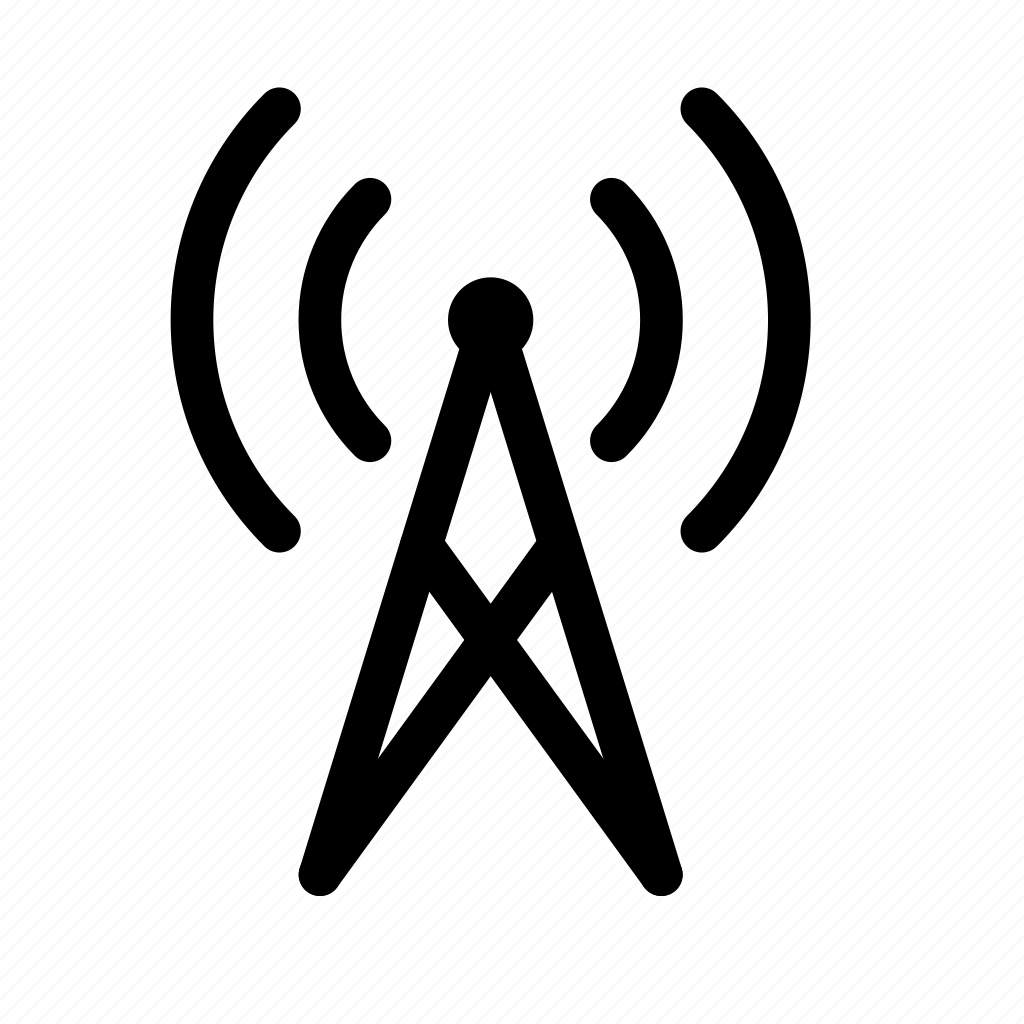}};
    \draw (-0.375, 0.5) node[above] {\footnotesize ETA};
  \end{tikzpicture}}%
  \caption{Scenario of the two-vehicle car-following problem.
  At $t_{0}$, the considered vehicle is at $w_{0}$ and the preceding vehicle at $w_{1}$.
  The preceding vehicle travels through a set of way points $(w_{1}, w_{2}, \dots, w_{f})$, where each pair of consecutive way points are spaced $\Delta{s}$ apart, barring boundary condition.}
  \label{fig:problem_setup}
\end{figure}

For the considered vehicle to properly follow the preceding vehicle, we impose the following five requirements.
First and foremost, to prevent collision, we constrain the space headway between the two vehicles to be greater than some minimal headway $h_{\min}(t): \mathbb{R} \rightarrow \mathbb{R}_{\geq 0}$.
Next, to \textit{follow} the preceding vehicle, we constrain the space headway to be smaller than some maximum headway $h_{\max}(t): \mathbb{R} \rightarrow \mathbb{R}_{\geq 0}$.
The third constraint requires the speed of the considered vehicle to fall within some speed limits $[v_{\min}, v_{\max}]$.
Next, we constrain the acceleration of the considered vehicle to be bounded within some range $[a_{\min}, a_{\max}]$.
Last but not least, among all trajectories that satisfies the constraints above, we select one that is the smoothest, as measured by $\ell^{2}$-norm of the acceleration.

We summarize the above specifications in Problem~\ref{prob:car_following}.

\begin{problem}
Provided $s^{0}(t)$ for $t \in [t_{0}, t_{f}]$ with some finite $t_{f}$, determine a trajectory of least acceleration in $\ell^{2}$-norm for the considered vehicle $s^{1}(t)$ for $t \in [t_{0}, t_{f}]$ such that \textsl{1}) $\min_{\forall t} s^{0}(t) - s^{1}(t) \geq h_{\min}(t)$, \textsl{2}) $\max_{\forall t} s^{0}(t) - s^{1}(t) \leq h_{\max}(t)$, \textsl{3}) $v^{1}(t) \in [v_{\min}, v_{\max}]$ for all $t$, and \text{4}) $a^{1}(t) \in [a_{\min}, a_{\max}]$ for all $t$.
\label{prob:car_following}
\end{problem}

With the problem stated above, we present the proposed controller in the next section.

\section{Optimal Car-Following Controller}
\label{sec:optimal_car-following_controller}

In this section, we present a hierarchical MPC scheme consisting of \text{1}) a prediction layer, \textsl{2}) a planning layer, and \textsl{3}) a tracking layer.
% The prediction layer provides the planning layer with an estimate of the trajectory of the preceding vehicle $\hat{s}^{0}(t)$ for a fixed spatial receding horizon.
% The planning layer computes an corresponding optimal acceleration plan for the considered vehicle $\check{v}^{1}(t)$ using $\hat{s}^{0}(t)$.
% The tracking layer tracks planned acceleration $\check{a}^{1}(t)$ while maintaining the considered vehicle inside some suitable headway constraints as often as possible.
% All layers are running iteratively in real time with different receding horizons.
The layered controller design is illustrated in Figure~\ref{fig:controller_design}.
We present the details of the three layers below.

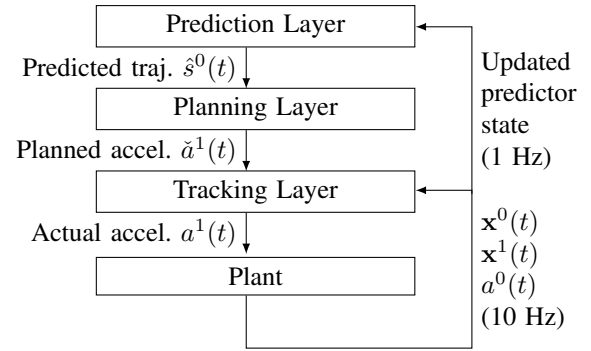
\begin{figure}[htbp]
  \centering
  \begin{tikzpicture}
    \node[draw, minimum width=4cm, minimum height=0.5cm, anchor=south west, text width=4cm, align=center] (predictor) at (0, 0) {Prediction Layer};
    \node[draw, minimum width=4cm, minimum height=0.5cm, anchor=south west, text width=4cm, align=center] (planner) at (0, -1.1) {Planning Layer};
    \node[draw, minimum width=4cm, minimum height=0.5cm, anchor=south west, text width=4cm, align=center] (controller) at (0, -2.2) {Tracking Layer};
    \node[draw, minimum width=4cm, minimum height=0.5cm, anchor=south west, text width=4cm, align=center] (plant) at (0, -3.3) {Plant};
    \draw[->, >=latex] (2, 0) to node[left] {Predicted traj. $\hat{s}^{0}(t)$} (2, -0.55);
    \draw[->, >=latex] (2, -1.1) to node[left] {Planned accel. $\check{a}^{1}(t)$} (2, -1.65);
    \draw[->, >=latex] (2, -2.2) to node[left] {Actual accel. $a^{1}(t)$} (2, -2.75);
    \draw[->, >=latex] (2, -3.3) to (2, -4) to (5, -4) to node[right, text width=0.5cm] {$\vec{x}^{0}(t)$ $\vec{x}^{1}(t)$ $a^{0}(t)$ (10~Hz)} (5, -1.9) to (4.25, -1.9);
    \draw[->, >=latex] (5, -1.95) to node[right, text width=1.5cm] {Updated predictor state (1~Hz)} (5, 0.275) to (4.25, 0.275);
  \end{tikzpicture}
  \caption{Diagram of the MPC controller design.
  %The layered controller scheme is designed to perform in receding horizon:
  Prediction layer and planning layer are invoked every 1 second.
  Tracking layer is triggered every 0.1 second.}
  \label{fig:controller_design}
\end{figure}

\subsection{Prediction layer}

The prediction layer predicts the trajectory of the preceding vehicle based on an ETA estimator.
Because the ETA estimator uses downstream information, it renders the controller non-local.
% Performance of the controller design is then evaluated with prediction layers of different levels of prediction resolutions and errors.
% In practice, such arrival time estimator is available from many existing map service providers, e.g. Google Maps and .
% In the following paragraphs, we describe how to predict the trajectory of the preceding vehicle based on a ETA estimator.

Define a way point every $\Delta{s}$ from $s^{0}(t)$ to $s^{0}(t) + l \cdot \Delta{s}$ at some time $t \in [t_{0}, t_{f}]$ and for some spatial receding horizon $l \cdot \Delta{s} > 0$.
Non-integer $l$ indicates that the spatial horizon cannot be evenly divided by $\Delta{s}$.
Collect all way points into a sequence, we have $\mathcal{W}(t) \coloneqq (w_{1}(t), w_{2}(t), \dots, w_{l+1}(t)) = ( s^{0}(t) + i \cdot \Delta{s} \mid i = 0, 1, \dots, l )$.
For convenience, we assume $w_{l+1} \leq w_{f}$.
% \begin{equation}
% \begin{aligned}
%     &\mathcal{W}(t) \coloneqq\\
%     &\resizebox{.9\hsize}{!}{$\begin{cases}
%       \left(s^{0}(t) + i_{w} \cdot \Delta{s} \mid i_{w} = 0, 1, \dots, l \right), & \text{if } l \leq \frac{w_{f} - s^{0}(t)}{\Delta{s}},\\
%       \left(s^{0}(t), s^{0}(t) + \cdot \Delta{s}, \dots, s^{0}(t) + \bar{l} \cdot \Delta{s}, w_{f}\right), & \text{otherwise},
%   \end{cases}$}
% \end{aligned}
%   \label{eq:waypoint_set}
% \end{equation}
% where the truncation index is $\bar{l} \coloneqq \lfloor \frac{w_{f} - s^{1}(t)}{\Delta{s}} \rfloor$.
% where
% \begin{equation*}
%   k =
%   \begin{cases}
%     \lfloor \frac{L}{\Delta{s}} \rfloor,& \text{if } \lfloor \frac{L}{\Delta{s}} \rfloor < \frac{L}{\Delta{s}},\\
%     \frac{L}{\Delta{s}} - 1,              & \text{otherwise}.
%   \end{cases}
% \end{equation*}
% For convenience, we write $f$ and $l + 1$ interchangeably.
% By construction, the preceding vehicle travels through the way point set $\mathcal{W}$ from $t_{0}$ to $t_{f}$.

Define $\widehat{\mathcal{T}}^{0}(t)$ as the corresponding ETA set for the preceding vehicle, that is, $\widehat{\mathcal{T}}^{0}(t) \coloneqq (t_{1}^{0}(t), \hat{t}_{2}^{0}(t), \dots, \hat{t}_{l+1}^{0}(t))$,
% \begin{equation*}
%   \widehat{\mathcal{T}}^{0}(t) \coloneqq
%   \begin{cases}
%       (t_{1}^{0}, \hat{t}_{2}^{0}, \dots, \hat{t}_{l}^{0}), & \text{if } l \leq \frac{w_{f} - s^{1}(t)}{\Delta{s}},\\
%       (t_{1}^{0}, \hat{t}_{2}^{0}, \dots, \hat{t}_{\bar{l}}^{0}, \hat{t}_{f}), & \text{otherwise},
%   \end{cases}
% \end{equation*}
where $t_{i}^{0}, \hat{t}_{i}^{0}$ are the \textit{true} and \textit{estimated} arrival times of the preceding vehicle at the way point $w_{i} \in \mathcal{W}(t)$.
% and $\hat{t}_{f}$ is the estimated final time.
% Of course, $t_{1}^{0} \coloneqq t$.
% for all $i_{w}$, respectively.
For the problem to be well-posed, we require both true and estimated arrival time sequences to be strictly increasing.
% The receding horizon design is illustrated in the $s$ axis in Figure~\ref{fig:receding_horizon_spec}.

% For the reminder of the article, we assume that $l \leq \frac{w_{f} - s^{1}(t)}{\Delta{s}}$, without any loss of generality.

We assume there exists an ETA estimator\footnote{Common ETA service providers include Google Maps, Waze, INRIX, and Garmin.} $f: \mathcal{W}(t) \rightarrow \widehat{\mathcal{T}}^{0}(t)$ that estimates the \textit{arrival} time of the preceding vehicle at each way point in $\mathcal{W}(t)$.
To model prediction errors, we assume that the relative estimation error is \textit{uniformly} distributed with an average at one and a radius of $\sigma$, that is,
\begin{equation*}
\frac{\hat{t}_{i+1}^{0} - \hat{t}_{i}^{0}}{t_{i+1}^{0} - t_{i}^{0}} \sim \mathcal{U}(1-\sigma, 1+\sigma), \quad i = 1, 2, \dots, l,
\end{equation*}
where $\hat{t}_{1}^{0} = t_{1}^{0}$.
We use $f_{\sigma}$ to denote that an ETA estimator $f$ has an error radius of $\sigma$.
% Parameters of the prediction layer are therefore $l, \Delta{s}$, and $\sigma$.

With $\mathcal{W}(t)$ and $\widehat{\mathcal{T}}^{0}(t) = f_{\sigma}(\mathcal{W}(t))$, we can then generate a \textit{predicted} trajectory for the preceding vehicle $\hat{s}^{0} : [t_{1}^{0}, \hat{t}_{l}^{0}] \rightarrow [s^{0}(t), s^{0}(t) + l \cdot \Delta{s}]$ using a standard interpolation method.
% that maps a given time $t$ to an estimated future position $\hat{s}^{0}(t)$.
The above receding horizon design is illustrated in the $s$ and $\hat{t}^{0}$ axes of Figure~\ref{fig:receding_horizon_spec}.

\begin{figure}[t]
  \centering
  \resizebox{0.49\textwidth}{!}{%
  \begin{tikzpicture}
    \draw[->, >=latex, line width=0.35mm] (0, 0.6) to (7, 0.6) node[right] {$s$}; % (res. $\Delta{s}$)};
    \draw (0, 0.7) to (0, 0.6) node[below] {$w_{1}$};
    \draw (1.25, 0.7) to (1.25, 0.6) node[below] {$w_{2}$};
    \draw (2.5, 0.7) to (2.5, 0.6) node[below] {$w_{3}$};
    \draw (3.75, 0.7) to (3.75, 0.6) node[below] {$w_{4}$};
    \draw (5.125, 0.6) node[below] {$\dots$};
    \draw (6.5, 0.7) to (6.5, 0.6) node[below] {$w_{l+1}$};

    \draw[->, >=latex, line width=0.35mm] (0, 0) to (7, 0) node[right] {$\hat{t}^{0}$};
    \draw (0, 0.1) to (0, 0) node[below] {$t_{1}^{0}$};
    \draw (1.75, 0.1) to (1.75, 0) node[below] {$\hat{t}_{2}^{0}$};
    \draw (3.00, 0.1) to (3.00, 0) node[below] {$\hat{t}_{3}^{0}$};
    \draw (4.00, 0.1) to (4.00, 0) node[below] {$\hat{t}_{4}^{0}$};
    \draw (5.25, 0) node[below] {$\dots$};
    \draw (6.5, 0.1) to (6.5, 0) node[below] {$\hat{t}_{l+1}^{0}$};

    \draw[->, >=latex, line width=0.35mm] (0, -0.7) to (5.5, -0.7) node[right] {$t^{p}$}; % (res. $\Delta{t}^{p}$)};
    \draw (0, -0.6) to (0, -0.7) node[below] {$t^{p}_{0}$};
    \draw (0.75, -0.6) to (0.75, -0.7) node[below] {$t^{p}_{1}$};
    \draw (1.5, -0.6) to (1.5, -0.7) node[below] {$t^{p}_{2}$};
    \draw (2.25, -0.6) to (2.25, -0.7) node[below] {$t^{p}_{3}$};
    \draw (3, -0.6) to (3, -0.7) node[below] {$t^{p}_{4}$};
    \draw (4, -0.7) node[below] {$\dots$};
    \draw (5, -0.6) to (5, -0.7) node[below] {$t^{p}_{m}$};

    \draw[->, >=latex, line width=0.35mm] (0, -1.4) to (3.7, -1.4) node[right] {$t^{c}$}; % (res. $\Delta{t}^{c}$)};
    \draw (0, -1.3) to (0, -1.4) node[below] {$t^{c}_{0}$};
    \draw (0.2, -1.3) to (0.2, -1.4) node[below] {};
    \draw (0.4, -1.3) to (0.4, -1.4) node[below] {$t^{c}_{1}$};
    \draw (0.6, -1.3) to (0.6, -1.4) node[below] {};
    \draw (0.8, -1.3) to (0.8, -1.4) node[below] {$t^{c}_{3}$};
    \draw (1.0, -1.3) to (1.0, -1.4) node[below] {};
    \draw (1.2, -1.3) to (1.2, -1.4) node[below] {$t^{c}_{5}$};
    \draw (1.4, -1.3) to (1.4, -1.4) node[below] {};
    \draw (1.6, -1.3) to (1.6, -1.4) node[below] {$t^{c}_{7}$};
    \draw (1.8, -1.3) to (1.8, -1.4) node[below] {};
    \draw (2.0, -1.3) to (2.0, -1.4) node[below] {$t^{c}_{9}$};
    \draw (2.6, -1.4) node[below] {$\dots$};
    \draw (3.2, -1.3) to (3.2, -1.4) node[below] {$t^{c}_{n}$};
  \end{tikzpicture}}%
  \caption{Specification of the receding-horizon control scheme.
  The prediction layer has a \textit{fixed} spatial horizon and a \textit{variable} temporal horizon.
  The planning layer has a temporal horizon $m$.
  The tracking layer has a fixed temporal horizon $n$.
  }
  \label{fig:receding_horizon_spec}
\end{figure}
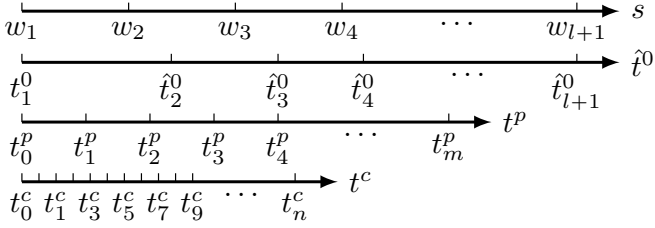

\subsection{Planning layer}

The planning layer plans a reference trajectory for the considered vehicle based on the predicted trajectory of the preceding vehicle.
% In this following paragraphs, we propose a planning layer design as a LCQP.

% As required by Problem~\eqref{prob:car_following}, we first define a minimum and a maximum headway envelopes for the problem.
With $\hat{s}^{0}(t)$ predicted by $f_{\sigma}$, we can derive the following four trajectories:
\textsl{1}) $\hat{s}_{t}^{-}(t) \coloneqq \hat{s}^{0}(t - \Delta{t}^{-})$: predicted envelope of some minimal time headway $\Delta{t}^{-}$,
\textsl{2}) $\hat{s}_{t}^{+}(t) \coloneqq \hat{s}^{0}(t - \Delta{t}^{+})$: predicted envelope of some maximal time headway $\Delta{t}^{+}$,
\textsl{3}) $\hat{s}_{s}^{-}(t) \coloneqq \hat{s}^{0}(t) - \Delta{s}^{-}$: predicted envelope of some minimal space headway $\Delta{s}^{-}$, and
\textsl{4}) $\hat{s}_{s}^{+}(t) \coloneqq \hat{s}^{0}(t) - \Delta{s}^{+}$: predicted envelope of some maximal space headway $\Delta{s}^{+}$.
The above four trajectories are illustrated in Figure~\ref{fig:envelope_design}.

While many possible designs are possible, in this study we select the minimal and maximal headway envelopes as follows
\begin{equation}\begin{aligned}
  \hat{s}_{\min}\left( t \mid \mathcal{W}, \widehat{\mathcal{T}}^{0} \right) &\coloneqq \max\left( \min\left( \hat{s}_{s}^{-}(t), \hat{s}_{t}^{-}(t) \right), \hat{s}_{s}^{+}(t) \right),\\
  \hat{s}_{\max}\left( t \mid \mathcal{W}, \widehat{\mathcal{T}}^{0} \right) &\coloneqq \min\left( \max\left( \hat{s}_{s}^{+}(t), \hat{s}_{t}^{+}(t) \right), \hat{s}_{s}^{-}(t) \right).
\end{aligned}
\label{eq:headway_constraints}
\end{equation}
The area between $\hat{s}_{\min}(t)$ and $\hat{s}_{\max}(t)$ is shaded in green in Figure~\ref{fig:envelope_design}.
Consequently, the minimum and maximum headways are
\begin{equation}\begin{aligned}
  h_{\min}(t) &= \hat{s}^{0}(t) - \hat{s}_{\min}\left( t \mid \mathcal{W}, \widehat{\mathcal{T}}^{0} \right),\\
  h_{\max}(t) &= \hat{s}^{0}(t) - \hat{s}_{\max}\left( t \mid \mathcal{W}, \widehat{\mathcal{T}}^{0} \right).\\
\end{aligned}\end{equation}
The objective of the design in~\eqref{eq:headway_constraints} is to generate smooth car-following without falling far behind nor getting too close.
Compared to~\cite{schmied2015nonlinear,jia2018energy}, our formulation is stricter, which leads to a tighter yet still safe and smooth car-following.

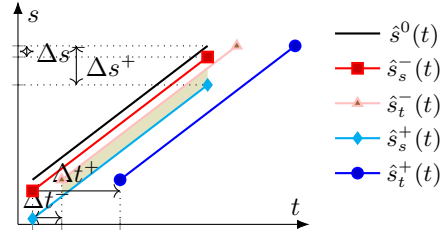
\begin{figure}[t]
  \centering
  \begin{tikzpicture}
    \begin{axis}[
        axis lines=middle,
        axis line style={-Latex},
        xmin=0,
        xmax=100,
        ymin=0,
        ymax=100,
        xlabel=$t$,
        ylabel=$s$,
        xtick={5,15,35},
        xticklabels={},
        ytick={62.5,75,80},
        yticklabels={},
        legend style={at={(1.5,0.125)},anchor=south east},
        legend style={draw=none,font=\small},
        width=0.3\textwidth,
        height=0.25\textwidth,
      ]
      \addplot+[no marks,domain=5:65,samples=100,black,thick] {x + 15};
      \addplot+[mark=square*,domain=5:65,samples=2,red,thick] {x + 10};
      \addplot+[mark=triangle*,domain=15:75,samples=2,pink,thick,name path=lower] {(x-10) + 15};
      \addplot+[mark=diamond*,domain=5:65,samples=2,cyan,thick,name path=upper] {x - 2.5};
      \addplot+[mark=*,domain=35:95,samples=2,blue,thick] {(x-30) + 15};
      \addplot[mark=none,black,dotted] coordinates {(5,0) (5,20)};
      \addplot[mark=none,black,dotted] coordinates {(15,0) (15,20)};
      \addplot[mark=none,black,dotted] coordinates {(35,0) (35,20)};
      \addplot[mark=none,black,dotted] coordinates {(0,80) (65,80)};
      \addplot[mark=none,black,dotted] coordinates {(0,75) (65,75)};
      \addplot[mark=none,black,dotted] coordinates {(0,62.5) (65,62.5)};
      \draw [<->] (5,3) -- node[above] {$\Delta{t}^{-}$} (15,3);
      \draw [<->] (5,15) -- node[above] {$\Delta{t}^{+}$} (35,15);
      \draw [<->] (3,75) -- node[right] {$\Delta{s}^{-}$} (3,80);
      \draw [<->] (20,62.5) -- node[right] {$\Delta{s}^{+}$} (20,80);
      \legend{$\hat{s}^{0}(t)$, $\hat{s}^{-}_{s}(t)$, $\hat{s}^{-}_{t}(t)$, $\hat{s}^{+}_{s}(t)$, $\hat{s}^{+}_{t}(t)$};
      \addplot[olive,opacity=0.25] fill between[of=lower and upper, soft clip={domain=15:65}];
    \end{axis}
  \end{tikzpicture}
  \caption{Headway envelope design.
  $\hat{s}^{0}(\cdot)$ is the predicted trajectory of the proceeding vehicle.
  $\hat{s}^{-}_{s}(\cdot), \hat{s}^{+}_{s}(\cdot)$ are estimated minimum and maximum \textit{space} headway envelopes, respectively.
  $\hat{s}^{-}_{t}(\cdot), \hat{s}^{+}_{t}(\cdot)$ are estimated minimum and maximum \textit{time} headway envelopes, respectively.
  The shaded region represents the feasible set in space and time.
  }
  \label{fig:envelope_design}
\end{figure}

% Hence, the predicted maximum and minimum allowable headway are $\hat{h}_{\max}(t) = -\hat{s}_{\max}\left( t \mid \mathcal{W}, \widehat{\mathcal{T}}^{0} \right) + \hat{s}^{0}(t), \hat{h}_{\min}(t) = \hat{s}_{\min}\left( t \mid \mathcal{W}, \widehat{\mathcal{T}}^{0} \right) - \hat{s}^{0}(t)$, respectively.

% With headway envelopes defined above, we are ready to define the main decision variables of the LCQP.
Now, let $m$ be the temporal planning horizon and $\Delta{t}_{p}$ be the temporal planning resolution.
For all $i_{p} = 0, \dots, m$, denote $t_{i_{p}}^{p} \coloneqq t^{0}_{1} + i_{p}\Delta{t}_{p}$.
Non-integer $m$ indicates that $\hat{t}_{m}^{0} - t^{0}_{1}$ cannot be evenly divided by $\Delta{t}_{p}$.
% Without loss of generality, we assume $t_{m}^{p} = \hat{t}_{l}^{0}$.
% \begin{equation*}
%   m \coloneqq
%   \begin{cases}
%     \lfloor \frac{\hat{t}_{l}^{0} - t}{\Delta{t}^{p}} \rfloor, & \text{if } \lfloor \frac{\hat{t}_{l}^{0} - t}{\Delta{t}^{p}} \rfloor < \frac{\hat{t}_{l}^{0} - t}{\Delta{t}^{p}}\\
%     \lfloor \frac{\hat{t}_{l}^{0} - t}{\Delta{t}^{p}} \rfloor - 1, & \text{otherwise}.
%   \end{cases}
% \end{equation*}
% Furthermore, for convenience, let $t_{m+1}^{p} \coloneqq \hat{t}_{l}^{0}$.

Define the planned states $\widecheck{\mathcal{X}}^{1p}$ and the planned accelerations $\widecheck{\mathcal{U}}^{1p}$ of the considered vehicle as follows:
\begin{equation}
\begin{aligned}
  \widecheck{\mathcal{X}}_{m}^{1p} &\coloneqq \{\check{\vec{x}}^{1}(t_{0}^{p}), \check{\vec{x}}^{1}(t_{1}^{p}), \dots, \check{\vec{x}}^{1}(t_{m}^{p}) )\},\\
  \widecheck{\mathcal{U}}_{m-1}^{1p} &\coloneqq \{\check{a}^{1}(t_{0}^{p}), \check{a}^{1}(t_{1}^{p}), \dots, \check{a}^{1}(t_{m-1}^{p})\}.
\end{aligned}
\label{eq:decision_variables_planning}
\end{equation}
Here, we introduce the check accent $\check{x}$ to denote that a variable $x$ is a \textit{planned} decision variable.

With everything defined above, we plan for the states and accelerations of the considered vehicle via the following LCQP:

\begin{mini!}|s|[2]                   % mini! = minimize
    {X}                               % optimization variable
    {\resizebox{.75\hsize}{!}{${\displaystyle \alpha \sum_{i_{p}=0}^{m-1} \left(\check{a}^{1}(t_{i_{p}}^{p})\right)^{2} + \beta \sum_{j_{p}=1}^{m} \left( \xi_{j_{p}}^{p} \right)^{2} + \gamma \sum_{j_{p}=1}^{m} \left( \zeta_{j_{p}}^{p} \right)^{2}}$} \label{eq:planning-lcqp-obj}}   % objective function and label
    {\label{eq:planning-lcqp}}             % label for optimizatio problem
    {}                                % optimization result
    \addConstraint{\check{\vec{x}}^{1}(t_{0}^{p})}{ = \vec{x}^{1}_{0} \label{eq:planning-lcqp-con1}}
    \addConstraint{\check{\vec{x}}^{1}(t_{i_{p}+1}^{p})}{=
      \mat{A}_{i_{p}}^{p} \cdot \check{\vec{x}}^{1}(t_{i_{p}}^{p}) +
      \mat{B}_{i_{p}}^{p} \cdot \check{a}^{1}(t_{i_{p}}^{p}) \label{eq:planning-lcqp-con2}}
    \addConstraint{0 \leq \xi_{j_{p}}^{p}}{, \; \check{s}^{1}(t_{j_{p}}) - \hat{s}_{\min}\left( t_{j_{p}} \mid \mathcal{W}, \widehat{\mathcal{T}}^{0} \right) \leq \xi_{j_{p}}^{p} \label{eq:planning-lcqp-con3}}
    \addConstraint{0 \leq \zeta_{j_{p}}^{p}}{, \; -\check{s}^{1}(t_{j_{p}}) + \hat{s}_{\max}\left( t_{j_{p}} \mid \mathcal{W}, \widehat{\mathcal{T}}^{0} \right) \leq \zeta_{j_{p}}^{p}  \label{eq:planning-lcqp-con4}}
    \addConstraint{v_{\min} \leq \check{v}^{1}(t_{j_{p}}^{p})}{\leq v_{\max} \label{eq:planning-lcqp-con5}}
    \addConstraint{a_{\min} \leq \check{a}^{1}(t_{i_{p}}^{p})}{\leq a_{\max}, \label{eq:planning-lcqp-con6}}
\end{mini!}
where
\begin{equation*}\begin{aligned}
X &\coloneqq \left\{ \widecheck{\mathcal{X}}_{m}^{1p},\widecheck{\mathcal{U}}_{m-1}^{1p}, \{ \xi_{j_{p}}^{p} \}_{1}^{m}, \{ \zeta_{j_{p}}^{p} \}_{1}^{m} \right\},\\
\mat{A}_{i_{p}}^{p} &\coloneqq e^{\mat{A}(t_{i_{p}+1}^{p} - t_{i_{p}}^{p})}, \quad \mat{A} = \left[[0 \; 0]^{\top} \: [1 \; 0]^{\top}\right],\\
\mat{B}_{i_{p}}^{p} &\coloneqq \int^{t_{i_{p}+1}^{p}}_{t_{i_{p}}^{p}} e^{\mat{A}(t_{i_{p}+1}^{p} - \tau)} d\tau \cdot \mat{B}, \quad \mat{B} = [0 \; 1]^{\top},
\end{aligned}\end{equation*}
for $i_{p} = 0, \dots, m-1$ and $j_{p} = 1, \dots, m$.
The receding horizon design is illustrated in the $t^{p}$ axis of Figure~\ref{fig:receding_horizon_spec}.
We have $\mathcal{W}, \widehat{\mathcal{T}}^{0}$, $\alpha, \beta, \gamma$, and $\vec{x}_{0}^{1}$ as the inputs to the optimization program, where $\alpha, \beta, \gamma \in \mathbb{R}_{>0}, \alpha + \beta + \gamma = 1$, and $\vec{x}_{0}^{1}$ is the initial state of the considered vehicle.

In~\eqref{eq:planning-lcqp-obj}, the first term penalizes large accelerations, while the second and third term, coupled with~\eqref{eq:planning-lcqp-con3} and~\eqref{eq:planning-lcqp-con4}, regulate the vehicle inside the admissible space headway ranges.
Constraints~\eqref{eq:planning-lcqp-con1} and~\eqref{eq:planning-lcqp-con2} impose kinematic constraints using zero-order hold.
Lastly, constraints~\eqref{eq:planning-lcqp-con5} and~\eqref{eq:planning-lcqp-con6} impose speed and acceleration limits.

\begin{proposition}
If the feasible set is nonempty, problem~\eqref{eq:planning-lcqp} has unique global optimum.
\label{prop:opt}
\end{proposition}
\begin{proof}
Because $\alpha, \beta, \gamma > 0$, the objective function~\eqref{eq:planning-lcqp-obj} is strongly convex.
By construction, the feasible set is a polyhedron, which is also convex.
Therefore, a direct application of Lemma 8.2 and Theorem 8.6 in \cite{calafiore2014optimization} reveals that if the feasible set is nonempty, problem~\eqref{eq:planning-lcqp} has an unique global optimum.
\end{proof}

\begin{proposition}
For some initial time $t_{0}$, if the feasible set of problem~\eqref{eq:planning-lcqp} is nonempty, then its corresponding receding-horizon policy has a solution for all $t > t_{0}$, that is, it is persistently feasible.
\label{prop:feasibility}
\end{proposition}
\begin{proof}
To show that problem~\eqref{eq:planning-lcqp} is persistently feasible, it suffices to show that it has a solution at every future time step.
Because the feasible set is initially nonempty, the problem has a solution at $t_{0}$.
For every other time step $t > t_{0}$, it is easy to verify that zero acceleration is a feasible point due to the fact that constraints~\eqref{eq:planning-lcqp-con3} and~\eqref{eq:planning-lcqp-con4} are soft.
Therefore, it follows that problem~\eqref{eq:planning-lcqp-obj}, when repeatedly applied in receding horizon, always stays feasible.
\end{proof}

\begin{remark}
When the initial velocity is too high or too low, the feasible set of~\eqref{eq:planning-lcqp} will become empty.
For example, when the initial velocity is greater than $v_{\max}+a_{\max}\Delta{t}^{p}$, there will be no acceleration within the actuation limits that can steer the velocity below $v_{\max}$ in the next time step.
To prevent any chance of infeasibility, we can replace the hard constraints on velocity with suitable soft constraints.
\label{rem:inf}
\end{remark}

\begin{remark}
Because the constraints on minimum headway are soft, it is possible for the considered vehicle to violate constraint~\eqref{eq:planning-lcqp-con3}.
Therefore, it is theoretically permissible for the preceding vehicle to collide with the preceding vehicle.
Note that this is chosen by design because the hard constraint version of~\eqref{eq:planning-lcqp-con3} can never be guaranteed in practice due to irregularity of human driving and inevitable errors in prediction.
Therefore, the presented soft constraints in~\eqref{eq:planning-lcqp-con3} should be viewed as a ``best-effort'' attempt to respect the minimum headway requirement.
\label{rem:col}
\end{remark}

Furthermore, we briefly comment on how constraints~\eqref{eq:planning-lcqp-con3} and~\eqref{eq:planning-lcqp-con4} can be extended to account for cut-ins and cut-outs.
For example, a possible design is illustrated in Figure~\ref{fig:envelope_design_cutinandout}.
Under this design, the considered vehicle will not overreact when a discontinuous change occurs in the position of the preceding vehicle.
Rather, it gradually recovers to a comfortable headway over time.
With the help a switching condition, this policy can be combined with that of~\eqref{eq:headway_constraints} to form a hybrid MPC controller to perform car-following with cut-ins and cut-outs.

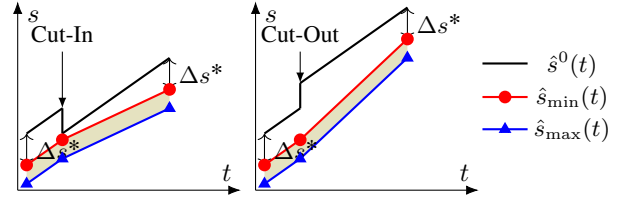
\begin{figure}[htbp]
\centering
  \begin{tikzpicture}
    \begin{axis}[
        axis lines=middle,
        axis line style={-Latex},
        xmin=0,
        xmax=125,
        ymin=0,
        ymax=150,
        xlabel=$t$,
        ylabel=$s$,
        ticks=none,
        xtick={},
        xticklabels={},
        ytick={},
        yticklabels={},
        legend style={at={(1,0.125)},anchor=south east},
        width=0.25\textwidth,
        height=0.225\textwidth,
      ]
      \addplot+[no marks,domain=5:25,samples=100,black,thick] {x + 40};
      \addplot+[no marks,domain=25:85,samples=100,black,thick] {x + 20};
      \addplot[mark=none,black,thick] coordinates {(25,65) (25,45)};
      \addplot[mark=*,red,thick,name path=lower] coordinates {(5,20) (25,40) (85,80)};
      \addplot[mark=triangle*,blue,thick,name path=upper] coordinates {(5,5) (25,25) (85,65)};
      \addplot[olive,opacity=0.25] fill between[of=lower and upper, soft clip={domain=5:85}];
      \draw [<->] (5,20) -- node[right] {\small $\Delta{s}^{*}$} (5,45);
      \draw [<->] (85,80) -- node[right] {\small $\Delta{s}^{*}$} (85,105);
      \draw [-latex] (25,110) node[above]{\small Cut-In} to (25,65);
    \end{axis}
  \end{tikzpicture}
  \begin{tikzpicture}
    \begin{axis}[
        axis lines=middle,
        axis line style={-Latex},
        xmin=0,
        xmax=125,
        ymin=0,
        ymax=150,
        xlabel=$t$,
        ylabel=$s$,
        ticks=none,
        xtick={},
        xticklabels={},
        ytick={},
        yticklabels={},
        legend style={at={(1.65,0.5)},anchor=east},
        legend style={draw=none,font=\small},
        width=0.25\textwidth,
        height=0.225\textwidth,
      ]
      \addplot+[no marks,domain=5:25,samples=100,black,thick] {x + 40};
      \addplot[mark=*,red,thick,name path=lower] coordinates {(5,20) (25,40) (85,120)};
      \addplot[mark=triangle*,blue,thick,name path=upper] coordinates {(5,5) (25,25) (85,105)};
      \addplot[olive,opacity=0.25] fill between[of=lower and upper, soft clip={domain=5:85}];
      \addplot+[no marks,domain=25:85,samples=100,black,thick] {x + 60};
      \addplot[mark=none,black,thick] coordinates {(25,65) (25,85)};
      \draw [<->] (5,20) -- node[right] {\small $\Delta{s}^{*}$} (5,45);
      \draw [<->] (85,120) -- node[right] {\small $\Delta{s}^{*}$} (85,145);
      \draw [-latex] (25,110) node[above]{\small Cut-Out} to (25,85);
      \legend{$\hat{s}^{0}(t)$, $\hat{s}_{\min}(t)$, $\hat{s}_{\max}(t)$};
    \end{axis}
  \end{tikzpicture}
  \caption{Headway envelope design to handle cut-ins and cut-outs.
  Left design is for cut-ins and right for cut-outs.
  The proposed designs aim to to keep a comfortable headway $\Delta{s}^{*}$ without overreaction.
  }
  \label{fig:envelope_design_cutinandout}
\end{figure}

% \Fangyu{Elaborate how it can be extended to multi vehicle cooperation}
% Last but not least, we comment on how to generalize the proposed controller to coordinate with other vehicles.
% Consider a platoon of $n$ vehicles that follow a preceding human vehicle.
% The first vehicle in the platoon works exactly as what is proposed in the article.
% After, it broadcasts the planned trajectory to the next vehicle platoon.
% The next vehicle, instead of using a prediction layer to estimate the future trajectory of its lead vehicle, can simply plan accordingly to the shared planned trajectory.
% We continue such plan and broadcast scheme until reaching the end of the platoon.

\subsection{Tracking layer}
To handle prediction errors, an additional tracking layer is introduced to track the planned accelerations and to guard the considered vehicle from imminent collisions.
% In this article, we implement such tracking by modifying the LCQP in~\eqref{eq:planning-lcqp}.

Concretely, the layer takes the following form, which is a modified version of~\eqref{eq:planning-lcqp}.
\begin{mini!}|s|[2]                   % mini! = minimize
    {Z}                               % optimization variable
    {\lambda \sum_{i_{c}=0}^{n-1} \left(\check{a}^{1}(t_{i_{c}}^{c}) - \bar{a}^{1}(t_{i_{c}}^{c})\right)^{2} + \mu \sum_{j_{c}=1}^{n} \left( \xi_{j_{c}}^{c} \right)^{2} \label{eq:tracking-lcqp-obj}}   % objective function and label
    {\label{eq:tracking-lcqp}}             % label for optimizatio problem
    {}                                % optimization result
    \addConstraint{\bar{\vec{x}}^{1}(t_{0}^{c})}{ = \vec{x}^{1}_{0} \label{eq:tracking-lcqp-con1}}
    \addConstraint{\bar{\vec{x}}^{1}(t_{i_{c}+1}^{c})}{=
      \mat{A}_{i_{c}}^{c} \cdot \bar{\vec{x}}^{1}(t_{i_{c}}^{c}) +
      \mat{B}_{i_{c}}^{c} \cdot \bar{a}^{1}(t_{i_{c}}^{c}) \label{eq:tracking-lcqp-con2}}
    \addConstraint{0 \leq \xi_{j_{c}}^{c}}{, \quad s^{1}(t_{j_{c}}^{c}) - \tilde{s}_{\min}\left( t_{j_{c}}^{c} \mid a^{0}(t_{0}^{c}) \right) \leq \xi_{j_{c}}^{c} \label{eq:tracking-lcqp-con3}}
    \addConstraint{v_{\min} \leq \bar{v}^{1}(t_{j_{c}}^{c})}{\leq v_{\max} \label{eq:tracking-lcqp-con4}}
    \addConstraint{a_{\min} \leq \bar{a}^{1}(t_{i_{c}}^{c})}{\leq a_{\max}, \label{eq:tracking-lcqp-con5}}
\end{mini!}
where
\begin{equation*}\begin{aligned}
t_{i_{c}}^{c} &\coloneqq t + i_{c}\Delta{t}^{c},\quad Z \coloneqq \left\{ \overline{\mathcal{X}}_{n}^{1c},\overline{\mathcal{U}}_{n-1}^{1c}, \{ \xi_{j_{c}}^{c} \}_{1}^{n} \right\},\\
\mat{A}_{i_{c}}^{c} &\coloneqq e^{\mat{A}(t_{i_{c}+1}^{c} - t_{i_{c}}^{c})}, \quad \mat{B}_{i_{c}}^{c} \coloneqq \int^{t_{i_{c}+1}^{c}}_{t_{i_{c}}^{c}} e^{\mat{A}(t_{i_{c}+1}^{c} - \tau)} d\tau \cdot \mat{B},
\end{aligned}\end{equation*}
for some horizon $n$, resolution $\Delta{t}^{c}$, $i_{c} = 0, \dots, n-1$, and $j_{c} = 1, \dots, n$.
% Clearly, we require $t_{n}^{c} \leq t_{m}^{p}$.
% Upon successful solution of~\eqref{eq:tracking-lcqp}, the acceleration of the considered vehicle from $t^{c}_{0}$ to $t^{c}_{1}$ is then set to $\bar{a}^{1}(t^{c}_{0})$ and the receding horizon is rolled forward by $\Delta{t}^{c}$.
$\overline{\mathcal{X}}_{n}^{1c},\overline{\mathcal{U}}_{n-1}^{1c}$ are defined similar to those in~\eqref{eq:decision_variables_planning}.
Note that here we use accent $\bar{x}$ to indicate that a variable $x$ is an internal variable of the tracking controller.
The receding horizon design is illustrated in the $t^{c}$ axis of Figure~\ref{fig:receding_horizon_spec}.
Unlike the planning layer, the minimum headway envelope $\tilde{s}_{\min}(\cdot \mid a^{0}(\cdot))$ is generated by assuming that the preceding vehicle accelerates \textit{constantly} from $t_{0}^{c}$ to $t_{n}^{c}$.
Similar to the planning layer once again, we have $a^{0}(t_{0})$, $\lambda, \mu$, and $\vec{x}_{0}^{1}$ as inputs to the optimization program, where $\lambda, \mu \in \mathbb{R}_{>0}, \lambda + \mu = 1$.

The modified LCQP in the tracking layer in~\eqref{eq:tracking-lcqp} is structurally identical to~\eqref{eq:planning-lcqp} except for:
\textsl{1}) we remove the soft constraint on maximum headway;
\textsl{2}) we adopt a shorter planning horizon $n$ and a smaller temporal resolution $\Delta{t}^{c}$; and
\textsl{3}) we predict the trajectory of the preceding vehicle by assuming that it keeps its current acceleration for the entire tracking horizon.

Note that the last change above is made to ensure safety.
Because constraints~\eqref{eq:tracking-lcqp-con3} only takes vehicular RADAR measurements, which are generally highly reliable, the collision avoidance of the car-following is mostly \textit{decoupled} from the ETA estimation.
Hence, a bad prediction from the ETA estimator cannot severely impact the safety of the vehicle.

\begin{remark}
Because the optimization problem~\eqref{eq:tracking-lcqp} may be viewed as a special case of problem~\eqref{eq:planning-lcqp} where $\gamma = 0$, it has the same properties as those highlighted in Propositions~\ref{prop:opt},\ref{prop:feasibility} and Remarks~\ref{rem:inf},\ref{rem:col}.
\end{remark}

\begin{remark}
We assume that we can directly control the acceleration of the vehicle.
Consequently, dynamic constraints~\eqref{eq:tracking-lcqp-con2} are purely \textit{kinematic}.
\end{remark}

\added[id=FW]{
Lastly, we briefly describe how our proposed method can be extended for multi-vehicle platooning.
For example, consider a platoon of vehicles with indices $0, 1, \dots, k$, where index $0$ indicates the preceding human vehicle and indices $1, \dots, k$ indicate the $k$ following automated vehicles.
As before, we can use an ETA estimator to generate a prediction for vehicle zero.
For all the vehicles that \textit{follow}, we can substitute the predicted trajectories of the preceding vehicles by the corresponding planned trajectories.
That is, for $i = 1, 2, \dots, k$, we can use $\check{a}^{i-1}(t)$ to compute $\check{s}^{i-1}(t)$ and approximate $\hat{s}^{i-1}(t)$ with $\check{s}^{i-1}(t)$.
Integrating the above prediction method with the planning and tracking methods presented in this paper, we have consequently produced a complete platoon controller design.
}

\section{Numerical Simulation}
\label{sec:numerical_simulation}

To evaluate the proposed controller, we test its performance through an array of numerical experiments, each equipped with an ETA estimator of an unique combination of spatial resolution $\Delta{s}$ and noise level $\sigma$.

For all simulations, we emulate the preceding vehicle by replaying a recorded drive from~\cite{matthew_nice_2021_6366762}.
The velocity and acceleration of the recorded drive are shown in the second and third rows of Figure~\ref{fig:planning}.

Spatial horizon of the ETA estimators are fixed to $l \cdot \Delta{s} = \SI{3000}{\meter}$.
We run an array of simulations by varying the spatial resolution $\Delta{s}$ between $\SI{10}{\meter}$ and $\SI{500}{\meter}$ and noise level $\sigma$ between $0.01$ and $0.25$.
To establish a performance baseline, we run an additional simulation with an intelligent driver model (IDM) controller~\cite{treiber2000congested}.
Parameters of the IDM are chosen to roughly match that of the MPC controller: $a = 1.5, b = 3, \delta = 4, s_{0} = 3.5, \ell = 4.65$.
Likewise, to establish a performance upper bound, we run one more simulation with a controller that can perfectly foresee the trajectory of the preceding vehicle.
We call this controller the oracle controller.

Shared parameters of the planning and tracking layers are: $v_{\min} = 0, v_{\max} = \SI{35}{\meter\per\second}, a_{\min} = \SI{-1.5}{\meter\per\second\squared}, a_{\max} = \SI{3}{\meter\per\second\squared}, \Delta{s}^{-} = \SI{5}{\meter}, \Delta{s}^{+} = \SI{100}{\meter}, \Delta{t}^{-} = \SI{0.6}{\second}$, and $\Delta{t}^{+} = \SI{3.0}{\second}$.
Unique parameters of the planning layer are: $\alpha = 0.2, \beta = 0.7, \gamma = 0.1, \Delta{t}^{p} = \SI{1}{\second}$, and $m = 60$.
Unique parameters of the tracking layer are: $\lambda = 0.1, \mu = 0.9, \Delta{t}^{c} = \SI{0.1}{\second}$, and $n = 30$.

\added[id=FW]{
To evaluate performance, we use two measures $e$ and $f$ to quantify tracking error and fuel saving, respectively.
The measure $e$ is defined to be the standard deviation of $\check{v}^{1*} - \bar{v}^{1}$ for all sampled time instants, where $\check{v}^{1*}$ denotes the planned velocity under \textit{perfect} prediction.
The measure $f$ is defined to be the ratio of the fuel consumption of the baseline IDM controller to that of the MPC controller, using the energy model described in~\cite{10.1145/3459609}.
}
% The results of the numerical simulations are shown in the fourth row of Table~\ref{tab:numerical_results}.

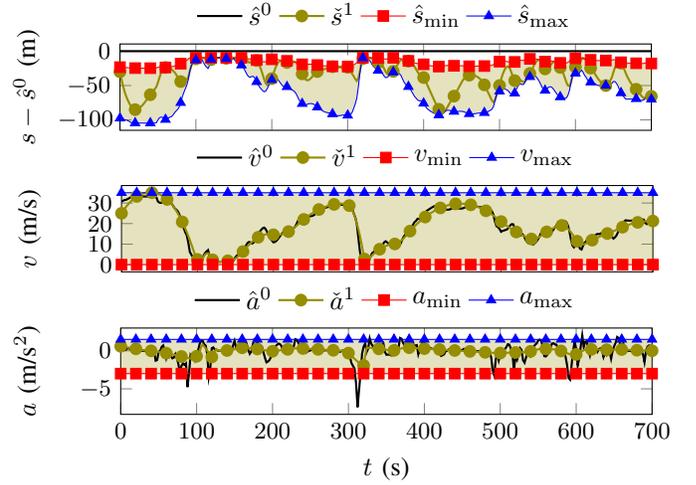
\begin{figure}[htbp]
  \centering
%   \begin{tabular}{ll}
%   \begin{subfigure}{0.39\textwidth}
% %   \hspace{1em}
%   \begin{tikzpicture}
%     \node[inner sep=0pt] (i24) at (0,3){\includegraphics[width=0.6\textwidth]{figures/i24.png}};
%     \node[inner sep=0pt] (rav4) at (0,0){\includegraphics[width=0.6\textwidth]{figures/rav4_sideview.jpeg}};
%     \draw[draw] (-0.5,1.5) rectangle ++(1.6,1.75);
%     \node[text width=3cm] at (1,3.5) {\footnotesize Trajectory Range};
%     \node[text width=4.5cm,black] at (0.25,1) {\footnotesize Toyota RAV4};
%   \end{tikzpicture}
%   \end{subfigure}&
%   \begin{subfigure}{0.59\textwidth}
  \hspace{-1.7em}
  \begin{tikzpicture}
      \begin{axis}[
          trim axis right,
          name=position_plot,
          axis line style={-Latex},
          xlabel={},
          xticklabels={},
          xmax=700,
          ylabel=$s - \hat{s}^{0}$ (m),
          width=0.475\textwidth,
          height=0.15\textwidth,
          legend style={draw=none,legend columns=-1,at={(0.5,1.6)},anchor=north},
          enlarge x limits=false,
          mark repeat=10,
        ]
        \addplot[no marks,thick] table[x=t, y expr=\thisrowno{1}-\thisrowno{1},col sep=comma] {data/7050_solution.csv};
        \addlegendentry{$\hat{s}^{0}$};
        \addplot[mark=*,olive,thick] table[x=t, y expr=\thisrowno{4}-\thisrowno{1},col sep=comma] {data/7050_solution.csv};
        \addlegendentry{$\check{s}^{1}$};
        \addplot[mark=square*,red,name path=lower] table[x=t, y expr=\thisrowno{7}-\thisrowno{1},col sep=comma] {data/7050_solution.csv};
        \addlegendentry{$\hat{s}_{\min}$};
        \addplot[mark=triangle*,blue,name path=upper] table[x=t, y expr=\thisrowno{8}-\thisrowno{1},col sep=comma] {data/7050_solution.csv};
        \addlegendentry{$\hat{s}_{\max}$};
        \addplot[olive,opacity=0.25] fill between[of=upper and lower];
      \end{axis}
  \end{tikzpicture}

  \vspace{-0.75em}
  \hspace{-1.45em}
  \begin{tikzpicture}
      \begin{axis}[
          name=position_plot,
          axis line style={-Latex},
          xlabel={},
          xticklabels={},
          xmax=700,
          ylabel=$v$ (m/s),
          width=0.475\textwidth,
          height=0.15\textwidth,
          legend pos=outer north east,
          legend style={draw=none,legend columns=-1,at={(0.5,1.6)},anchor=north},
          enlarge x limits=false,
          mark repeat=10,
        ]
        \addplot[no marks,thick] table[x=t, y expr=\thisrowno{2},col sep=comma] {data/7050_solution.csv};
        \addlegendentry{$\hat{v}^{0}$};
        \addplot[mark=*,olive,thick] table[x=t, y expr=\thisrowno{5},col sep=comma] {data/7050_solution.csv};
        \addlegendentry{$\check{v}^{1}$};
        \addplot[mark=square*,red,name path=lower,samples=100] table[x=t, y=v_min,col sep=comma] {data/7050_solution.csv};
        \addlegendentry{$v_{\min}$};
        \addplot[mark=triangle*,blue,name path=upper,samples=100] table[x=t, y=v_max,col sep=comma] {data/7050_solution.csv};
        \addlegendentry{$v_{\max}$};
        \addplot[olive,opacity=0.25] fill between[of=upper and lower];
      \end{axis}
  \end{tikzpicture}

  \vspace{-0.5em}
  \hspace{-0.95em}
  \begin{tikzpicture}
      \begin{axis}[
          name=position_plot,
          axis line style={-Latex},
          xlabel=$t$ (s),
          xmax=700,
          ylabel=$a$ (m/s$^2$),
          width=0.475\textwidth,
          height=0.15\textwidth,
          legend style={draw=none,legend columns=-1,at={(0.5,1.6)},anchor=north},
          enlarge x limits=false,
          mark repeat=10,
        ]
        \addplot[no marks,thick] table[x=t, y expr=\thisrowno{3},col sep=comma] {data/7050_solution.csv};
        \addlegendentry{$\hat{a}^{0}$};
        \addplot[mark=*,olive,thick] table[x=t, y expr=\thisrowno{6},col sep=comma] {data/7050_solution.csv};
        \addlegendentry{$\check{a}^{1}$};
        \addplot[mark=square*,red,name path=lower,samples=100] table[x=t, y=a_min,col sep=comma] {data/7050_solution.csv};
        \addlegendentry{$a_{\min}$};
        \addplot[mark=triangle*,blue,name path=upper,samples=100] table[x=t, y=a_max,col sep=comma] {data/7050_solution.csv};
        \addlegendentry{$a_{\max}$};
        \addplot[olive,opacity=0.25] fill between[of=upper and lower];
      \end{axis}
  \end{tikzpicture}
%   \end{subfigure}
%   \end{tabular}
  \caption{Demonstration of the planning layer.
           The oracle controller assumes perfect prediction over the entire planning horizon.}
%   The planned trajectory has significantly attenuated acceleration compared to that of the preceding vehicle.}
  \label{fig:planning}
\end{figure}

\begin{figure}[htbp]
  \centering
%   \hspace{0.2em}
  \begin{tikzpicture}
    \begin{axis}[
      colormap/Reds,
      colorbar,
      colorbar style={
        ylabel=$e$ (m/s),
        yticklabel style={
            align=left,
        },
        width=0.1*\pgfkeysvalueof{/pgfplots/parent axis width}
      },
      view={0}{90},
      width=0.235\textwidth,
      height=0.23\textwidth,
      xtick = {0.01,0.05,0.1,0.15,0.2,0.25},
      xticklabels = {0.01,0.05,0.1,0.15,0.2,0.25},
      x tick label style={rotate=60,anchor=east},
      xlabel=$\sigma$ (s),
      ytick = {10,100,200,300,400,500},
      ylabel=$\Delta{s}$,
      ]
      \addplot3[surf, shader=flat] table[x index=0,y index=1,z index=2, col sep=comma] {data/error_stdev.txt};
    \end{axis}
  \end{tikzpicture}
  \begin{tikzpicture}
    \begin{axis}[
      colormap/Reds,
      colorbar,
      colorbar style={
        ylabel=$f$ (\%),
        yticklabel style={
            align=left,
        },
        width=0.1*\pgfkeysvalueof{/pgfplots/parent axis width}
      },
      view={0}{90},
      width=0.235\textwidth,
      height=0.23\textwidth,
      xlabel=$\sigma$ (s),
      ylabel=$\Delta{s}$ (m),
      xtick = {0.01,0.05,0.1,0.15,0.2,0.25},
      xticklabels = {0.01,0.05,0.1,0.15,0.2,0.25},
      x tick label style={rotate=60,anchor=east},
      xlabel=$\sigma$ (s),
      ytick = {10,100,200,300,400,500},
      ylabel=$\Delta{s}$,
      ]
      \addplot3[surf, shader=flat] table[x index=0,y index=1,z index=2, col sep=comma] {data/fuel_reduction.txt};
    \end{axis}
  \end{tikzpicture}
  \caption{\added[id=FW]{Performance heat maps of the MPC controllers.
  The top heat map measures tracking error $e$, while the bottom heat map measures fuel saving $f$.}}
  \label{fig:performance_heatmaps}
\end{figure}
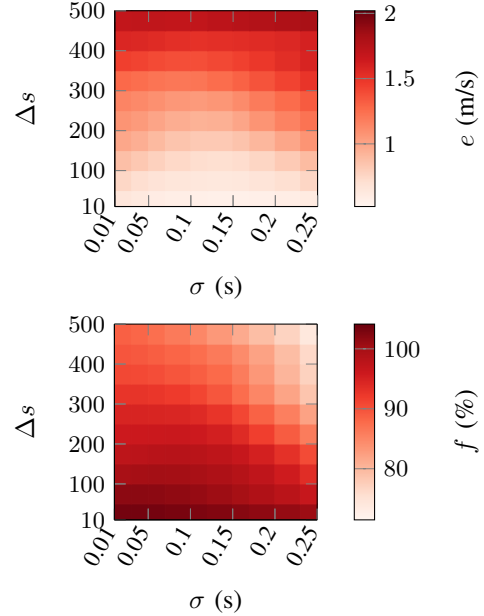

\added[id=FW]{
Position, velocity, and acceleration of the oracle controller are shown in Figure~\ref{fig:planning}.
As the oracle vehicle has demonstrated in Figure~\ref{fig:planning}, the optimal behavior is to avoid accelerating or decelerating unless there is an imminent violation of at least one of the two headway constraints.
% Table~\ref{tab:numerical_results} suggests that with a good ETA estimator the optimal trajectory can perform tight car-following with 20\% less fuel consumption compared to that of the preceding vehicle.

Tracking error and fuel saving are visualized in the top and bottom heat maps in Figure~\ref{fig:performance_heatmaps}.
As expected, tracking error $e$ is proportional to spatial resolution $\Delta{s}$ and noise level $\sigma$.
Likewise, fuel saving $f$ is \textit{inversely} proportional to spatial resolution $\Delta{s}$ and noise level $\sigma$.
Note that both performance measures platoon near the origin, indicating that the MPC controller is robust to prediction errors and especially so to that induced by $\sigma$.
}

Lastly, we comment on the running times of the planning and tracking layers.
The optimization solver is part of a python package called CVXOPT~\cite{andersen2020cvxopt}.
Benchmarked on an Intel i7-6700K CPU, the average running time of the planning layer is about 0.12 s and of the tracking layer about 0.04 s.
From Figure~\ref{fig:controller_design}, we find that both times fit comfortably into their allocated time budgets.

\section{Discussions}
\label{sec:discussions}

In this section, we discuss two of the most consequential features in the design of the controller, namely, \textsl{1}) the space between the maximum and minimum headway constraints and \textsl{2}) the resolution and noise level of the ETA estimators.

It is clear that the larger the space between the maximum and minimum headway constraints, the smoother the car-following could be.
For example, if the considered car is allowed to be arbitrarily far behind the preceding vehicle without incurring any penalties, it can simply wait for the preceding vehicle to exit the road and then accelerate to a very small constant velocity to complete the trip.
Nevertheless, we know from common sense that such behavior is not acceptable in most of the real-world applications.

On the contrary, when the space between the two headway constraints is small, the MPC controller will become sensitive to prediction errors.
Consider a situation where the allowable headway gap is thin and the prediction errors are significant.
Because the considered vehicle is already close to the boundaries of the headway constraints, an error in prediction could easily mislead the planning layer to falsely believe that it will soon violate one of the headway constraints.
In an effort to steer away from imminent constraint violation, the planning layer overreacts, leading to undesirably large acceleration or deceleration.

To resolve the above problem, one could modify the headway constraints proposed in~\eqref{eq:headway_constraints}.
One possible solution is to enforce a minimum space gap between the maximum and minimum headway envelopes and to add a small penalty to encourage the considered vehicle to drive at the center of the allowable headway constraints.

\added[id=FW]{
Last but not least, we comment on how to choose an ETA estimator.
As illustrated in Figure~\ref{fig:performance_heatmaps}, the performance of the controller changes significantly with the resolution and noise level of the ETA estimator.
In practice, one may first define a desirable tracking error $e_{\textnormal{des}}$ and a desirable fuel saving $f_{\textnormal{des}}$.
With $e_{\textnormal{des}}$ and $f_{\textnormal{des}}$ defined, one can then draw two level sets in the two heat maps of Figure~\ref{fig:performance_heatmaps}.
Any ETA estimator with $e$ and $f$ that simultaneously fall within the two level sets may be deemed sufficient to meet the design specifications.
}

\section{Conclusions}
\label{sec:conclusions}

In this article, we propose a hierarchical MPC control scheme based on a LCQP.
We show via simulations that the controller can achieve smooth and tight car-following and is robust to prediction errors.
Additional constructions can be added to further enhance its robustness to erroneous forecasts.
Possible future works include field tests with various ETA estimators, extension to handle cut-ins and cut-outs, and modification for platooning.

% \section*{Acknowledgment}

% The authors thank Alexander Keimer, Arwa AlAnqary, and Yiling You for their insightful comments and discussion.

\newpage
\bibliographystyle{unsrt}
\bibliography{references}

\end{document}